\documentclass[11pt]{article}

\usepackage{amsthm}
\usepackage{amssymb}
\usepackage{graphicx}
\usepackage{amsmath}
\usepackage{amsfonts}

\newtheorem{theo}{Theorem}
\newtheorem{propo}[theo]{Proposition}
\newcommand{\itl}{\textit}

\def\fpd#1#2{{\displaystyle\frac{\partial #1}{\partial #2}}}
\def\spd#1#2#3{{\displaystyle\frac{\partial^2 #1}
{\partial #2\partial #3}}}
\def\onehalf{{\frac12}}

\parskip=\medskipamount
\title{Hamiltonization of Nonholonomic Systems and
the Inverse Problem of the Calculus of Variations}
\author{A.M.\ Bloch$^{a}$\footnote{abloch@umich.edu},\, O.E.\ Fernandez$^{a}$\footnote{oscarum@umich.edu}\, and T.\
Mestdag$^{a,b}$\footnote{tom.mestdag@ugent.be}\\[2mm] {\small $^{a}$Department
of Mathematics, University of Michigan,}\\ {\small 530 Church
Street, Ann
Arbor, MI-48109, USA} \\[2mm] {\small $^{b}$Department of Mathematical Physics and
Astronomy, Ghent University,}\\ {\small Krijgslaan 281, S9, 9000
Gent, Belgium}}
\date{}

\begin{document}

\maketitle


\begin{abstract}

We introduce a method which allows one to recover the equations of
motion of a class of nonholonomic systems by finding instead an
unconstrained Hamiltonian system on the full phase space, and to
restrict the resulting canonical equations to an appropriate
submanifold of phase space. We focus first on the Lagrangian picture
of the method and deduce the corresponding Hamiltonian from the
Legendre transformation. We illustrate the method with several
examples and we discuss its relationship to the Pontryagin maximum
principle.

\end{abstract}




\section{Introduction}

The direct motivation of this paper lies with some interesting
results that appeared in the paper \cite{QB}, wherein the authors
propose a way to quantize some of the well-known classical examples
of nonholonomic systems. On the way to quantization, the authors
propose an alternative Hamiltonian representation of nonholonomic
mechanics. In short, the authors start off from the actual solutions
of the nonholonomic system, and apply a sort of Hamilton-Jacobi
theory to arrive at a Hamiltonian whose Hamilton's equations, when
restricted to a certain subset of phase space, reproduce the
nonholonomic dynamics. Needless to say, even without an explicit
expression for the solutions
 one can still derive a lot of the interesting
 geometric features and of the qualitative behaviour of a nonholonomic system.
 However, the ``Hamiltonization'' method introduced in \cite{QB} is not generalized to systems for which the
  explicit solution is not readily available, and hence cannot be applied to those systems.

In this paper, we wish to describe a method to Hamiltonize a class
of nonholonomic systems that does not depend on the knowledge of the
solutions of the system. Instead, we will start from the Lagrangian
equations of motion of the system and treat the search for a
Hamiltonian which Hamiltonizes the dynamics as the search for a
regular Lagrangian. That is, we will explain how one can associate
to the nonholonomic equations of motion a family of systems of
second-order ordinary differential equations and we will apply the
inverse problem of the calculus of variations \cite{Douglas,CV} on
those associated systems. If an unconstrained (or free) regular
Lagrangian exists for one of the associated systems, we can always
find an associated Hamiltonian by means of the Legendre
transformation. Since our method only makes use of the equations of
motion of the system, it depends only on the Lagrangian and
constraints of the nonholonomic system, but not on the knowledge of
the exact solutions of the system.

A system for which no exact solutions are known can only be
integrated by means of numerical methods. In addition to the above
mentioned application to quantization, our Hamiltonization method
may also be useful from this point of view. A geometric integrator
of a Lagrangian system uses a discrete Lagrangian that resembles as
close as possible the continuous Lagrangian (see e.g.\ \cite{West}).
On the other hand, the succes of a nonholomic integrator (see e.g.\
\cite{CortesMartinez,FedZen}) relies not only on the choice of
 a discrete Lagrangian but also on the choice of a discrete version of the constraint manifold. It seems therefore
reasonable that if a free Lagrangian for the nonholonomic system
exists, the Lagrangian integrator may perform better than a
nonholonomic integrator with badly chosen discrete constraints. Work
along these lines is in progress.

It should be remarked from the outset that the Hamiltonization we
have in mind is different from the ``Hamiltonization'' used in e.g.\
the papers \cite{Borisov,Ehlers}. Roughly speaking, these authors
first project a given nonholonomic system with symmetry to a system
on a reduced space and then use a sort of time reparametrization to
rewrite the reduced system in a Hamiltonian form in the new time
(this is the so-called Chaplygin's reducibility trick). In contrast,
we embed the (unreduced) nonholonomic system in a larger Hamiltonian
one.


In the second part of the paper, we show that in the cases where a
regular Lagrangian (and thus a Hamiltonian) exists, we can also
associate a first order controlled system to the nonholonomic
system. As an interesting byproduct of the method it turns out that
if one considers the optimal control problem of minimizing the
controls for an appropriate cost function under the constraint of
that associated first order controlled system, Pontryagin's maximum
principle leads in a straightforward way to the associated
Hamiltonians.

We begin with a quick review of nonholonomic mechanics in section 2,
where we introduce some of the well-known classical nonholonomic
systems which fall into the class of systems we will be studying in
the current paper. We then begin our investigations in section 3
with the Lagrangian approach to the problem. We detail the various
ways to associate a second-order system to a nonholonomic system,
which then forms the backbone of our subsequent analysis. In section
4 we briefly review the set up for the inverse problem of the
calculus of variations, and then apply it to some of the associated
second-order systems. We derive the corresponding Hamiltonians in
section 5 and discuss their relation with Pontryagin's maximum
principle in section 6.  At the end of the paper we provide a few
directions for future work on generalizing our findings to more
general nonholonomic systems, as well as applying them to quantize
nonholonomic systems.


\section{Nonholonomic systems}

Nonholonomic mechanics takes place on a configuration space $Q$ with
a nonintegrable distribution ${\cal D}$ that describes the (linear
supposed) kinematic constraints of interest. These constraints are
often given in terms of independent one-forms, whose vanishing in
turn describes the distribution. Moreover, one typically assumes
that one can find a fibre bundle and an Ehresmann connection $A$ on
that bundle such that ${\cal D}$ is given by the horizontal
subbundle associated with $A$. Such an approach is taken, for example, in some recent books on nonholonomic systems \cite{Bl,Co}.

Let $Q$ be coordinatized by coordinates $(r^I,s^\alpha)$, chosen in
such a way that the projection of the above mentioned bundle structure
is locally simply $(r,s) \mapsto r$. Moreover, let $\{\omega^{\alpha}\}$ be a
set of independent one-forms whose vanishing describes the
constraints on the system. Locally, we can write them as
\[
\omega^{\alpha}(r,s) =
ds^{\alpha}+A^{\alpha}_{I}(r,s)dr^{I}.
\]
The distribution ${\cal D}$ is then given by
\[ {\cal D} =
\text{span}\{\partial_{r^{I}}-A^{\alpha}_{I}\partial_{s^{\alpha}}\}.
\]
One then derives the equations of motion using the
Lagrange-d'Alembert principle, which takes into account the need for
reaction forces that enforce the constraints throughout the motion
of the system (see e.g.\ \cite{Bl}). If $L(r^\alpha,s^a,{\dot
r}^\alpha,{\dot s}^\alpha)$ is the Lagrangian of the system, these
equations are
\[
\frac{d}{dt}\Big(\fpd{L}{{\dot r}^I}\Big) - \fpd{L}{r^I} =
\lambda_\alpha A^\alpha_I\quad\mbox{and}\quad
\frac{d}{dt}\Big(\fpd{L}{{\dot s}^\alpha}\Big) - \fpd{L}{s^\alpha} =
\lambda_\alpha,
\]
together with the constraints ${\dot s}^\alpha = - A^\alpha_I {\dot
r}^{I}$. One can easily eliminate the Lagrange multipliers $\lambda$
and rewrite the above equations in terms of the constrained
Lagrangian
\[ L_{c}(r^I,s^\alpha,{\dot r}^I) =
L(r^I,s^\alpha,{\dot r}^I, - A^\alpha_I {\dot r}^{I}).
\]
The equations of motion, now in terms of $L_{c}$, become
\begin{equation}\label{nonhol}
\left\{
\begin{array}{rcl}
{\dot s}^\alpha &\!\!\!=\!\!\!& - A^\alpha_I {\dot r}^{I},
\\[2mm] \displaystyle \frac{d}{dt}\Big(\fpd{L_c}{{\dot r}^I}\Big) &\!\!\!=\!\!\!&
\displaystyle \fpd{L_c}{r^I} - A^\alpha_I \fpd{L_c}{s^\alpha} -
{\dot r}^J B^\alpha_{IJ} \fpd{L}{{\dot s}^\alpha}.
\end{array}
\right.,
\end{equation}
where $\displaystyle B^\alpha_{IJ} = \partial_{r^J}{A^\alpha_I} -
\partial_{r^I} {A^\alpha_J} + A^\beta_I \partial_{s^\beta}{A^\alpha_J} -
A^\beta_J \partial_{s^\beta}{A^\alpha_I}$.

To illustrate this formulation, consider perhaps the simplest
example: a nonholonomically constrained free particle with unit mass
moving in $\mathbb{R}^{3}$ (more details can be found in \cite{Bl},
\cite{R}). In this example one has a free particle with Lagrangian
and constraint given by \begin{equation}\label{nhfp}
L=\frac{1}{2}\left(\dot{x}^{2} + \dot{y}^{2} +
\dot{z}^{2}\right),\qquad\dot{z} +x\dot{y}=0.
\end{equation}
We can form the constrained Lagrangian $L_{c}$ by substituting the
constraint into $L$, and proceed to compute the constrained
equations, which take the form \begin{equation}\label{nhp1} \ddot x
=0,\quad \ddot y = -\frac{x\dot x\dot y}{1+x^2},\quad \dot z =-x
\dot y.
\end{equation}

Another example of interest is the knife edge on a plane. It
corresponds physically to a blade with mass $m$ moving in the $xy$ plane at an
angle $\phi$ to the $x$-axis (see \cite{NF}). The Lagrangian and
constraints for the system are:
\begin{equation}
 L = \frac{1}{2}m(\dot{x}^{2}+\dot{y}^{2})+\frac{1}{2}J\dot{\phi}^{2},
 \qquad
 \dot{x}\sin(\phi) - \dot{y}\cos(\phi) = 0, \label{kep1}
\end{equation}
from which we obtain the constrained equations:
\[
\ddot \phi =0,\qquad \ddot x = -\tan(\phi) \dot\phi \dot x,\qquad
\dot y = \tan(\phi)\dot x.
\]


\section{Second-order dynamics associated to a class of nonholonomic systems}

Recall from the introduction that we wish to investigate how we can
associate a free Hamiltonian to a nonholonomic system. One way to do
that is to rephrase the question in the Lagrangian formalism and to
first investigate whether or not there exists a regular Lagrangian.
Then, by means of the Legendre transformation, we can easily
generate the sought after Hamiltonian. Rather than abstractly describing the various ways of associating a
second-order system to a given nonholonomic system though, we will instead illustrate the method by means
 of one of the most interesting examples of a nonholonomic system.


\subsection{Associated Second-Order Systems for the vertically rolling disk}

The vertical rolling disk is a homogeneous disk rolling without slipping on a horizontal plane, with configuration space $Q=\mathbb{R}^{2}\times S^{1} \times S^{1}$ and parameterized by the coordinates $(x,y,\theta,\varphi)$, where $(x,y)$ is the position of the center of mass of the disk, $\theta$ is the angle that a point fixed on the disk makes with respect to the vertical, and $\varphi$ is measured from the positive $x$-axis. The
system has the Lagrangian and constraints given by
\begin{eqnarray}
L  &=& \frac{1}{2}m(\dot{x}^{2} + \dot{y}^{2}) + \frac{1}{2}I\dot{\theta}^{2} + \frac{1}{2}J\dot{\varphi}^{2}, \nonumber \\
 \dot{x} &=& R\cos(\varphi)\dot{\theta}, \nonumber \\
 \dot{y} &= &R\sin(\varphi)\dot{\theta}, \label{vd1}
\end{eqnarray}
where $m$ is the mass of the disk, $R$ is its radius, and $I,J$ are
the moments of inertia about the axis perpendicular to the plane of
the disk, and about the axis in the plane of the disk, respectively.
The constrained equations of motion are simply:
\begin{equation}\label{VRD}
\ddot\theta=0,\quad \ddot\varphi=0,\quad \dot x = R\cos(\varphi)
\dot\theta,\quad \dot y = R\sin(\varphi) \dot\theta.
\end{equation}
The solutions of the first two equations are of course
\[
\theta(t)  = u_{\theta}t + \theta_{0},\qquad \varphi(t)  =
u_{\varphi}t + \varphi_{0},
\]
and in the case where $u_\varphi \neq 0$, we get that the $x$- and
$y$-solution is of the form
\begin{eqnarray}
x(t)  &=& \left(\frac{u_{\theta}}{u_{\varphi}}\right)R\sin(\varphi(t)) + x_0,\nonumber \\
y(t)  &=&
-\left(\frac{u_{\theta}}{u_{\varphi}}\right)R\cos(\varphi(t)) + y_0,
\label{vd2}
\end{eqnarray}
from which we can conclude that the disk follows a circular path. If
$u_\varphi=0$, we simply get the linear solutions
\begin{equation}
x(t)  = R\cos(\varphi_0)u_\theta t + x_0, \quad y(t)  =
R\sin(\varphi_0)u_\theta t + y_0. \label{vd22}
\end{equation}
The situation in (\ref{vd22}) corresponds to the case when $\varphi$
remains constant, i.e. when the disk is rolling along a straight
line. For much of what we will discuss in the next sections, we will
exclude these type of solutions from our framework for reasons we
discuss later.

Having introduced the vertical disk, let us take a closer look at
the nonholonomic equations of motion (\ref{VRD}). As a system of
ordinary differential equations, these equations form a mixed set of
coupled first- and second-order equations. It is well-known that
these equations are never variational on their own \cite{Bl,Co}, in
the sense that we can never find a regular Lagrangian whose
(unconstrained) Euler-Lagrange equations are equivalent to the
nonholonomic equations of motion (\ref{nonhol}) (although it may
still be possible to find a singular Lagrangian). There are,
however, infinitely many systems of second-order equations (only),
whose solution set contains the solutions of the nonholonomic
equations (\ref{nonhol}). We shall call these second-order systems
{\em associated second-order systems}, and in the next section will
wish to find out whether or not we can find a regular Lagrangian for
one of those associated second-order systems. If so, we can use the
Legendre transformation to get a full Hamiltonian system on the
associated phase space. On the other hand, the Legendre
transformation will also map the constraint distribution onto a
constraint submanifold in phase space. The nonholonomic solutions,
considered as particular solutions of the Hamiltonian system, will
then all lie on that submanifold.

There are infinitely many ways to arrive at an associated
second-order system for a given nonholonomic system. We shall
illustrating three choices below using the vertical rolling disk as
an example.

Consider, for example, taking the time derivative of the constraint
equations, so that a solution of the nonholonomic system (\ref{VRD})
also satisfies the following complete set of second-order
differential equations in all variables $(\theta,\varphi,x,y)$:
\begin{equation}\label{choice1VRD}
\ddot\theta=0,\quad \ddot\varphi=0,\quad \ddot x = -R\sin(\varphi)
\dot\theta\dot\varphi,\quad \ddot y = R\cos(\varphi)
\dot\theta\dot\varphi.
\end{equation}
We shall call this associated second-order system the {\em first
associated second-order system}. Excluding for a moment the case
where $u_\varphi=0$, the solutions of equations (\ref{choice1VRD})
can be written as
\begin{eqnarray*}
\theta(t)  &= & u_{\theta}t + \theta_{0}\\ \varphi(t)  &=&
u_{\varphi}t + \varphi_{0}\\
x(t)  &=& \left(\frac{u_{\theta}}{u_{\varphi}}\right)R\sin(\varphi(t)) + u_xt+ x_0,\nonumber \\
y(t)  &=&
-\left(\frac{u_{\theta}}{u_{\varphi}}\right)R\cos(\varphi(t)) +
u_yt+ y_0.
\end{eqnarray*}
By restricting the above solution set to those that also satisfy the
constraints $\dot x = \cos(\varphi)\dot\theta$ and $\dot y =
\sin(\varphi)\dot\theta$ (i.e. to those solutions above with
$u_x=u_y=0$), we get back the solutions (\ref{vd2}) of the
non-holonomic equations (\ref{VRD}). A similar reasoning holds for
the solutions of the form (\ref{vd22}). The question we then wish to
answer in the next section is whether the second-order equations
(\ref{choice1VRD}) are equivalent to the Euler-Lagrange equations of
some regular Lagrangian or not.

Now, taking note of the special structure of equations (\ref{choice1VRD}), we may use the
constraints (\ref{VRD}) to eliminate the $\dot{\theta}$ dependency. This yields another plausible choice for
an associated system:
\begin{equation}\label{choice2VRD}
\ddot\theta=0,\quad \ddot\varphi=0,\quad \ddot x =
-\frac{\sin(\varphi)}{\cos(\varphi)} \dot x\dot\varphi,\quad \ddot y
= \frac{\cos(\varphi)}{\sin(\varphi)} \dot y \dot\varphi.
\end{equation}
We shall refer to this choice later as the {\em second associated
second-order system}.

Lastly, we may simply note that, given that on the constraint
manifold the relation $\sin(\varphi)\dot x - \cos(\varphi)\dot y=0$
is satisfied, we can easily add a multiple of this relation to some
of the equations above. One way of doing so leads to the system
\begin{eqnarray}
J \ddot\varphi &=& -mR (\sin(\varphi)\dot x - \cos(\varphi)\dot y)
\dot \theta ,\nonumber\\ (I+mR^2)\ddot\theta &=& mR
(\sin(\varphi)\dot x - \cos(\varphi)\dot y) \dot \varphi,\nonumber\\
(I+mR^2)\ddot x & = & -R (I+mR^2)\sin(\varphi) \dot\theta\dot\varphi
 + mR^2 \cos(\varphi)
(\sin(\varphi)\dot x - \cos(\varphi)\dot y) \dot \varphi,\nonumber\\
(I+mR^2)\ddot y &=& R(I+mR^2)\cos(\varphi) \dot\theta\dot\varphi
 + mR^2 \sin(\varphi) (\sin(\varphi)\dot x -
\cos(\varphi)\dot y) \dot \varphi.\label{choice3VRD}
\end{eqnarray}
For later discussion we shall refer to it as the {\em third
associated second-order system}. We mention this particular
second-order system here because it has been shown in \cite{FB}
(using techniques that are different than those we will apply in
this paper) that this complicated looking system is indeed
variational! The Euler-Lagrange equations for the regular Lagrangian
\begin{equation}\label{LagFB}
L  = -\onehalf m ({\dot x}^2 + {\dot y}^2) + \onehalf I
{\dot\theta}^2 + \onehalf J {\dot\varphi}^2 +
mR\dot\theta(\cos(\varphi) \dot x+\sin(\varphi) \dot y),
\end{equation}
are indeed equivalent to equations (\ref{choice3VRD}), and, when
restricted to the constraint distribution, its solutions are exactly
those of the nonholonomic equations (\ref{VRD}). We shall have more
to say about this system in section 4.4 below.


\subsection{Associated Second-Order Systems in General}

We will, of course, not only be interested in the vertically rolling
disk. It should be clear by now that there is no systematic way to
catalogue the second-order systems that are associated to a
nonholonomic system. If no regular Lagrangian exists for one
associated system, it may still exist for one of the infinitely many
other associated systems. For many nonholonomic systems, the search
for a Lagrangian may therefore remain inconclusive. On the other
hand, also the solution of the inverse problem of any given
associated second-order system is too hard and too technical to
tackle in the full generality of the set-up of the section 2.
Instead, we aim here to concisely formulate our results for a
well-chosen class of nonholonomic systems which include the
aforementioned examples and for only a few choices of associated
second-order systems.

To be more precise, let us assume from now on that the configuration
space $Q$ is locally just the Euclidean space ${\mathbb R}^n$ and
that the base space of the fibre bundle is two dimensional, writing $(r_1,r_2;s_\alpha)$ for the coordinates. We will consider the class of nonholonomic systems where the Lagrangian is given by

\begin{equation}\label{LEucl}
L=\onehalf(I_1{\dot r_{1}}^2+I_2{\dot r_{2}}^2 + \sum_\alpha I_\alpha {\dot
s}_\alpha^2),
\end{equation}
(with all $I_\alpha$ positive constants) and where the constraints
take the following special form
\begin{equation}\label{constraints}
{\dot s}_\alpha =-A_\alpha(r_{1})\dot r_{2}.
\end{equation}
Although this may seem to be a very thorough simplification,  this
interesting class of systems does include, for example, all the
classical examples described above. We also remark that all of the
above systems fall in the category of so-called Chaplygin systems
(see \cite{Bl}). The case of 2-dimensional distributions was also
studied by Cartan, be it for other purposes (see e.g. \cite{Bryant}
and the references therein).

In what follows, we will assume that none of the $A_\alpha$ are
constant (in that case the constraints are, of course, holonomic).
The nonholonomic equations of motion (\ref{nonhol}) are now
\begin{equation}\label{NHeq}
\ddot r_{1} =0, \quad \ddot r_{2} = -N^2 \big(\sum_\beta I_\beta
A_\beta A'_\beta\big) \dot r_{1}\dot r_{2},\quad {\dot s}_\alpha
=-A_\alpha\dot r_{2},
\end{equation}
where $N$ is shorthand for the function
\begin{equation}\label{N} N(r_{1})=\frac{1}{\sqrt{I_2+\sum_\alpha
I_\alpha A_\alpha^2}}.
\end{equation}
This function is directly related to the invariant measure of the
system. Indeed, we have shown in \cite{FB} that for a two-degree of
freedom system such as (\ref{NHeq}), we may compute the density $N$
of the invariant measure (if it exists) by integrating two
first-order partial differential equations derived from the
condition
 that the volume form be preserved along the nonholonomic flow. In the present case, these two
 equations read:
\begin{equation} \frac{1}{N}\frac{\partial N}{\partial r_{1}} +
\frac{\sum_\beta I_\beta A_\beta A'_\beta}{I_2+\sum_\alpha I_\alpha
A_\alpha^2}=0, \qquad \frac{1}{N}\frac{\partial N}{\partial r_{2}} =
0, \label{im1}
\end{equation}
and obviously the expression for $N$ in (\ref{N}) is its solution up
to an irrelevant multiplicative constant. In case of the free
nonholonomic particle and the knife edge the invariant measure
density is $N\sim 1/\sqrt{1+x^2}$ and $N\sim
1/\sqrt{(1+\tan^2(\phi))} = \cos(\phi)$, respectively. In case of
the vertically rolling disk it is a constant. We shall see later
that systems with a constant invariant measure (or equivalently,
with constant $\sum_\alpha I_\alpha A_\alpha^2$) always play a
somehow special role.

We are now in a position to generalize the associated second-order systems presented in section 2.1 to the more general
class of nonholonomic systems above. In the set-up above, the first associated second-order system is, for the more general
systems (\ref{NHeq}), the system
\[ \ddot r_{1} =0, \quad \ddot r_{2} = -N^2 \big(\sum_\beta I_\beta A_\beta A'_\beta\big) \dot r_{1}\dot r_{2},\quad
{\ddot s}_\alpha = -(A'_{\alpha}\dot r_{1}\dot r_{2} + A_\alpha\ddot r_{2}),
\]
or equivalently, in normal form,
\begin{eqnarray}
\label{choice1}  &&\ddot r_{1} =0, \qquad \qquad \ddot r_{2} = -N^2
\big(\sum_\beta I_\beta A_\beta A'_\beta\big) \dot r_{1}\dot r_{2},\nonumber\\
&& {\ddot s}_\alpha = -\Big(A'_{\alpha} - N^2 A_\alpha
\big(\sum_\beta I_\beta A_\beta A'_\beta\big)\Big) \dot r_{1} \dot
r_{2}.
\end{eqnarray}
For convenience, we will often simply write
\[ \ddot r_{1} =0,\quad \ddot r_{2} =
\Gamma_2(r_{1})\dot r_{1}\dot r_{2},\quad {\ddot s}_\alpha =
\Gamma_\alpha(r_{1})\dot r_{1} \dot r_{2},
\]
for these types of second-order systems.

The second associated second-order system we encountered for the
vertically rolling disk also translates to the more general setting. We
get
\begin{eqnarray}
&& \ddot r_{1} =0, \qquad\qquad \ddot r_{2} = -N^2 \big(\sum_\beta I_\beta A_\beta A'_\beta\big) \dot r_{1}\dot r_{2},\nonumber\\
&& {\ddot s}_\alpha =  \Big(A'_{\alpha} - N^2 A_\alpha
\big(\sum_\beta I_\beta A_\beta A'_\beta\big)\Big) \dot r_{1}
\left(\frac{{\dot s}_\alpha}{A_\alpha}\right), \label{NHsode}
\end{eqnarray}
where in the right-hand side of the last equation, there is no sum
over $\alpha$. A convenient byproduct of this way of associating a
second-order system to (\ref{NHeq}) is that now all equations
decouple except for the coupling with the $r_{1}$-equation. To
highlight this, we will write this system as
\begin{equation} \ddot
r_{1} =0, \quad \ddot q_a = \Xi_a(r_{1}) \dot q_a\dot r_{1}
\nonumber
\end{equation}
(no sum over $a$) where, from now on, $(q_a)=(r_{2},s_\alpha)$ and
$(q_i)=(r_{1},q_a)$.

We postpone the discussion about the third associated second-order
system of our class until section 4.4.


\section{Lagrangians for associated second-order systems}


\subsection{The inverse problem of Lagrangian mechanics}

Let $Q$ be a manifold with local coordinates $(q^i)$ and assume we are given a
system of second-order ordinary differential equations ${\ddot
q}^i=f^i(q,\dot q)$ on $Q$. The search for a regular Lagrangian
 is known in the literature as `the inverse problem of the calculus of variations,' and has a
 long history (for a recent survey on this history, see e.g. \cite{KP} and the long list of references therein). In order for a regular Lagrangian
$L(q,\dot{q})$ to exist we must be able to find functions
$g_{ij}(q,\dot q)$, so-called multipliers, such that
\[
g_{ij}({\ddot q}^j-f^j) = \frac{d}{dt}\left( \fpd{L}{{\dot q}^i}
\right)-\fpd{L}{q^i}.
\]
It can be shown \cite{Anderson,Douglas,CV} that the multipliers must
satisfy
\begin{eqnarray*}
&&\det(g_{ij})\neq 0,\quad\quad g_{ji}=g_{ij},\quad\quad
\fpd{g_{ij}}{{\dot q}^k}=\fpd{g_{ik}}{{\dot q}^j};\\&&
\Gamma(g_{ij}) - \nabla^k_j g_{ik}- \nabla^k_i g_{kj}=0,\\ &&
g_{ik}\Phi^k_j = g_{jk}\Phi^k_i;
\end{eqnarray*}
where $\nabla^i_j = -\onehalf
\partial_{{\dot q}^j}f^i$ and
\[
\Phi^k_j = \Gamma\left(\partial_{{\dot
q}^j}{f^k}\right)-2\partial_{q^j}{f^k}-\onehalf\partial_{{\dot
q}^j}{f^l}\partial_{{\dot q}^l}{f^k}.
\]
The symbol $\Gamma$ stands for the vector field ${\dot
q}^i\partial_{q^i} + f^i
\partial_{{\dot q}^i}$ on $TQ$ that can naturally be associated to the system ${\ddot
q}^i=f^i(q,\dot q)$. Conversely, if one can find functions $g_{ij}$
satisfying these conditions then the equations $\ddot{q}^i=f^i$ are
derivable from a regular Lagrangian. Moreover, if a regular
Lagrangian $L$ can be found, then its Hessian $\spd{L}{{\dot
q}^i}{{\dot q}^j}$ is a multiplier.

The above conditions are generally referred to as the Helmholtz
conditions. We will fix from the start $g_{ij}=g_{ji}$ for $j\leq
i$, and we will simply write $g_{ijk}$ for $\partial_{{\dot
q}^k}{g_{ij}}$, and also assume the notation to be symmetric over
all its indices.

The Helmholtz conditions are a mixed set of coupled algebraic and
PDE conditions in $(g_{ij})$. We will refer to the penultimate
condition as the `$\nabla$- condition,' and to the last one as the
`$\Phi$-condition.' The algebraic $\Phi$-conditions are of course
the most interesting to start from. In fact, we can easily derive
more algebraic conditions (see e.g.\ \cite{Towards}). For example,
by taking a $\Gamma$-derivative of the $\Phi$-condition, and by
replacing $\Gamma(g_{ij})$ everywhere by means of the
$\nabla$-condition, we arrive at a new algebraic condition of
the form
\[
g_{ik}(\nabla\Phi)^k_j = g_{jk}(\nabla\Phi)^k_i,
\]
where $(\nabla\Phi)^i_j = \Gamma(\Phi^i_j)  -
\nabla^i_m\Phi^m_j-\nabla^m_j\Phi^i_m$. As in \cite{Towards}, we
will call this new condition the $(\nabla\Phi)$-condition. It will,
of course, only give new information as long as it is independent
from the $\Phi$-condition (this will not be the case, for example,
if the commutator of matrices $[\Phi,\nabla\Phi]$ vanishes). One can
repeat the above process on the $(\nabla\Phi)$-condition, and so on
to obtain possibly independent
$(\nabla\ldots\nabla\Phi)$-conditions.

A second route to additional algebraic conditions arises from the
derivatives of the $\Phi$-equation in $\dot q$-directions. One can
sum up those derived relations in such a way that the terms in
 $g_{ijk}$ disappear on account of the symmetry in all their indices.
The new algebraic relation in $g_{ij}$ is then of the form
\[
g_{ij}R^j_{kl} + g_{lj}R^j_{ik} + g_{kj}R^j_{li}= 0,
\]
where $R^j_{kl}= \partial_{{\dot q}^j}(\Phi^k_i)-
\partial_{{\dot q}^i}(\Phi^k_j)$. For future use, we will call this the $R$-condition.

As before, this process can
 be continued to obtain more algebraic conditions. Also, any mixture of
the above mentioned two processes leads to possibly new and
independent algebraic conditions. Once we have used up all the
information that we can obtain from this infinite series of
algebraic conditions, we can start looking at the partial
differential equations in the $\nabla$-conditions.

We are now in a position to investigate whether a Lagrangian exists for the
two choices of associated systems (\ref{choice1}) and (\ref{NHsode}).


\subsection{Lagrangians for the first associated second-order system}

The first second-order system of interest is of the form
\begin{equation}\label{system1}
\ddot r_{1} =0,\quad \ddot r_{2} = \Gamma_2(r_{1})\dot r_{1}\dot r_{2},\quad {\ddot
s}_\alpha = \Gamma_\alpha(r_{1})\dot r_{1} \dot r_{2}.
\end{equation}
The only non-zero components of $(\Phi^i_j)$ are
\begin{eqnarray*}
&&\Phi^2_1 = (\onehalf \Gamma_2^2 - \Gamma'_2) \dot r_{1} \dot r_{2}, \qquad
\Phi^2_2 = - (\onehalf \Gamma_2^2 - \Gamma'_2) {\dot r_{1}}^2,\\ &&
 \Phi^\alpha_1 = (\onehalf \Gamma_\alpha \Gamma_2 - \Gamma_\alpha') \dot r_{1} \dot r_{2} ,\qquad \Phi^\alpha_2 =
 - (\onehalf \Gamma_\alpha \Gamma_2 - \Gamma_\alpha') {\dot r_{1}}^2.
\end{eqnarray*}
For $\nabla\Phi$ and $\nabla\nabla\Phi$ we get
\begin{eqnarray*} &&
(\nabla\Phi)^2_1 = (\Gamma_2\Gamma_2'-\Gamma_2'') {\dot r_{1}}^2 \dot
r_{2},\qquad (\nabla\Phi)^2_2 = - (\Gamma_2\Gamma_2'-\Gamma_2'') {\dot
 r_{1}}^3,\\&&
 (\nabla\Phi)^\alpha_1 = (\Gamma_\alpha\Gamma_2 -\Gamma_\alpha'') {\dot r_{1}}^2 \dot r_{2}, \qquad (\nabla\Phi)^\alpha_2 = - (\Gamma_\alpha\Gamma_2 -\Gamma_\alpha'') {\dot
 r_{1}}^3,
\end{eqnarray*}
and
\begin{eqnarray*}
(\nabla\nabla\Phi)^2_1 &= &
((\Gamma_2')^2+\Gamma_2\Gamma_2''-\Gamma_2''') {\dot r_{1}}^3 \dot r_{2},\\
(\nabla\nabla\Phi)^2_2 &= & -
((\Gamma_2')^2+\Gamma_2\Gamma_2''-\Gamma_2''') {\dot
 r_{1}}^4\\
 (\nabla\nabla\Phi)^\alpha_1 &= & (\Gamma_\alpha'\Gamma_2' + \frac{3}{2} \Gamma_\alpha\Gamma''_2 - \frac{1}{2} \Gamma_\alpha''\Gamma_2- \Gamma_\alpha''') {\dot r_{1}}^3 \dot r_{2},
 \\ (\nabla\nabla\Phi)^\alpha_2 &= & - (\Gamma_\alpha'\Gamma_2' + \frac{3}{2} \Gamma_\alpha\Gamma''_2 - \frac{1}{2} \Gamma_\alpha''\Gamma_2- \Gamma_\alpha''') {\dot
 r_{1}}^4,
\end{eqnarray*}
and so on.

We can already draw some immediate consequences just by looking at
the above explicit expressions. Let's make things a bit more
accessible by considering the case where the dimension is 4. Then,
the $\Phi$-equations of the system (\ref{system1}) and their
derivatives are all of the form
\begin{eqnarray}&& g_{12}\Psi^2_2 +
g_{13}\Psi^3_2 + g_{14}\Psi^4_2 = g_{22}\Psi^2_1 + g_{23}\Psi^3_1 +
g_{24}\Psi^4_1, \nonumber
\\&&g_{23}\Psi^2_1 + g_{33}\Psi^3_1 + g_{34}\Psi^4_1 =0,\nonumber
\\&& g_{23}\Psi^2_2 + g_{33}\Psi^3_2 +
g_{34}\Psi^4_2 =0,\label{Psi}
\\&& g_{24}\Psi^2_1 + g_{34}\Psi^3_1 +
g_{44}\Psi^4_1 =0,\nonumber
\\&& g_{24}\Psi^2_2 + g_{34}\Psi^3_2 +
g_{44}\Psi^4_2 =0,\nonumber
\end{eqnarray}
where, within the same equation, $\Psi$ stands for either $\Phi$,
$\nabla\Phi$, $\nabla\nabla\Phi$, $\nabla\nabla\nabla\Phi$, ... We
will refer to the equations of the first line in (\ref{Psi}) as
`equations of the first type,' and to equations of the next four
lines as `equations of the second type.' The first 3 equations of
the first type, namely those for $\Phi$, $\nabla\Phi$ and
$\nabla\nabla\Phi$ are explicitly:
\begin{eqnarray}&& g_{12}\Phi^2_2 +
g_{13}\Phi^3_2 + g_{14}\Phi^4_2 = g_{22}\Phi^2_1 + g_{23}\Phi^3_1 +
g_{24}\Phi^4_1,\nonumber\\&& g_{12}(\nabla\Phi)^2_2 +
g_{13}(\nabla\Phi)^3_2 + g_{14}(\nabla\Phi)^4_2
=g_{22}(\nabla\Phi)^2_1+ g_{23}(\nabla\Phi)^3_1 +
g_{24}(\nabla\Phi)^4_1,\label{pnpnnp}\\ &&
g_{12}(\nabla\nabla\Phi)^2_2 + g_{13}(\nabla\nabla\Phi)^3_2 +
g_{14}(\nabla\nabla\Phi)^4_2 = g_{22}(\nabla\nabla\Phi)^2_1 +
g_{23}(\nabla\nabla\Phi)^3_1 +
g_{24}(\nabla\nabla\Phi)^4_1.\nonumber
\end{eqnarray}
For the systems at hand, the particular expression of $\Phi$ and its
derivatives are such that
\begin{eqnarray*}
&&\Phi^2_2 (\nabla\Phi)^2_1 - \Phi^2_1 (\nabla\Phi)^2_2 =0,\\&&
(\nabla\Phi)^2_2 (\nabla\nabla\Phi)^2_1 - (\nabla\Phi)^2_1
(\nabla\nabla\Phi)^2_2 =0,
\end{eqnarray*}
and so on. By taking the appropriate linear combination of the first
and the second, and of the second and the third equation in
(\ref{pnpnnp}), we can therefore obtain two equations in which the
unknowns $g_{12}$ and $g_{22}$ are eliminated. Moreover, under
certain regularity conditions, these two equations can be solved for
$g_{13}$ and $g_{14}$ in terms of $g_{23}$ and $g_{24}$ (we will
deal with exceptions later on). So, if we can show that $g_{23}$ and
$g_{24}$ both vanish, then so will also $g_{13}$ and $g_{14}$. Then,
in that case $g_{12}\Psi^2_2 = g_{22}\Psi^2_1$, but no further
relation between $g_{12}$ and $g_{22}$ can be derived from this type
of algebraic conditions.

The infinite series of equations given by those of the second type
in (\ref{Psi}) are all equations in the 5 unknowns $g_{23}$,
$g_{33}$, $g_{34}$, $g_{24}$ and $g_{44}$. Not all of these
equations are linearly independent, however. In fact, given that the
system (\ref{system1}) exhibits the property
\[
\Psi^a_1\Psi^b_2 - \Psi^b_1\Psi^a_2 =0,
\]
(where $\Psi$ is one of $\Phi,\nabla\Phi,\nabla\nabla\Phi,...$), one
can easily deduce that the last four lines of equations in
(\ref{Psi}) actually reduce to only two kinds of equations. If we
assume that we can find among this infinite set 5 linearly
independent equations, there will only be the zero solution
\[
g_{23}=g_{33}=g_{34}=g_{24}=g_{44}=0,
\]
and from the previous paragraph we know that then also
$g_{14}=g_{13}=0$. To conclude, under the above mentioned
assumptions, the matrix of multipliers
\[
(g_{ij})=\left(\begin{array}{cccc} g_{11} & g_{12} & 0& 0
\\ g_{12} & g_{22} & 0 &0 \\ 0 & 0 & 0 & 0\\ 0 & 0 & 0 & 0
\end{array} \right)
\]
is singular and we conclude that there is no regular Lagrangian for
the system. The above reasoning can, of course, be generalized to
lower and higher dimensions.

We will refer to the above as `the general case'. The assumptions
made above are, however, not always satisfied, and they need to be
checked for every particular example. Let us consider first the
example of the (three-dimensional) nonholonomic particle, where
$\Gamma_2 = -x /(1+x^2)$ and $\Gamma_3 = -1/(1+x^2)$. The equations
for $\Psi=\Phi,\nabla\Phi$ of the second type give the following two
linear independent equations
\[
({\dot x}^2-2)g_{23} + 3x g_{33}=0,\quad (x^3 - 5x)g_{23} +
(5x^2-1)g_{33}=0.
\]
We can easily conclude that $g_{23}=g_{33}=0$. With that, the first
two equations of the first type are
\begin{eqnarray*}
&& (x^2-2)\dot x g_{12} + 3 x \dot x g_{13} + (x^2-2)\dot y
g_{22}=0,\\ && (x^3-5x)\dot x g_{12} + (5x^2-1) \dot x g_{13} +
(x^3-5x)\dot y g_{22}.
\end{eqnarray*}
From this $g_{13}=0$ and $\dot x g_{12}=-\dot y g_{22}$, and there
is therefore no regular Lagrangian.

With a similar reasoning (but with different coefficients) we reach
the same conclusion for the example of the knife edge on a plane.

The vertically rolling disk is a special case, however, and so is
any system (\ref{NHeq}) with the property that $\sum_\alpha I_\alpha
A_\alpha^2$ is a constant. This last relation is in fact equivalent
with the geometric assumption that the density of the invariant
measure $N$ is constant. In that case, we get $\Gamma_2=0$. Not only
does $\Gamma_2$ vanish, but so do all $\Psi^2_1$ and $\Psi^2_2$  for
$\Psi=\Phi,\nabla\Phi,...$. We also have $\Gamma_3 =
-R\sin(\varphi)$ and $\Gamma_4 = R\cos(\varphi)$. Moreover, looking
again first at expressions (\ref{pnpnnp}), one can easily show that
for the vertically rolling disk these three equations, and in fact
any of the equations that follow in that series, are all linearly
depending on the following two equations
\begin{eqnarray*}
&&\cos(\varphi) {\dot \varphi} g_{13} + \sin(\varphi) \dot \varphi
g_{14} + \cos(\varphi) \dot \theta g_{23} + \sin(\varphi)\dot \theta
g_{24} = 0,
\\&& \sin(\varphi) {\dot \varphi} g_{13} - \cos(\varphi) \dot \varphi g_{14} + \sin(\varphi) \dot
\theta g_{23} - \cos(\varphi)\dot \theta g_{24} = 0.
\end{eqnarray*}
Although these equations are already in a form where $g_{12}$ and
$g_{22}$ do not show up, it is quite inconvenient that there is no
way to relate these two unknowns to any of the other unknowns.
However, as in the general case, we can deduce from this an
expression for $g_{13}$ and $g_{14}$ as a function of $g_{23}$ and
$g_{24}$. We get
\begin{equation} g_{13} = -\frac{\dot \theta}{\dot
\varphi}g_{23},\qquad g_{14} = -\frac{\dot \theta}{\dot
\varphi}g_{24}.\label{help}
\end{equation}

The infinite series of equations of the second type (i.e.\ the last
four lines in (\ref{Psi})) are all linearly dependent to either one of
the following four equations
\begin{eqnarray*}
&&\cos(\varphi)g_{33} + \sin(\varphi) g_{34}=0,\quad
\cos(\varphi)g_{34} + \sin(\varphi) g_{44}=0\\ &&
\sin(\varphi)g_{33} -\cos(\varphi) g_{34}=0,\quad
\sin(\varphi)g_{34} - \cos(\varphi) g_{44}=0,
\end{eqnarray*}
from which $g_{33}=g_{34}=g_{44}=0$ follows immediately. In
comparison to the general case, however, we can no longer conclude
from the above that also $g_{23}$ and $g_{24}$ vanish, and
therefore, we can also not conclude from (\ref{help}) that $g_{13}$
and $g_{14}$ vanish. This concludes, in fact, the information we can
extract from the $\Phi$-condition, and the algebraic conditions that
follow from taking its derivatives w.r.t.\ $\nabla$. Also, any
attempt to create new algebraic conditions by means of the tensor
$R$ is fruitless, since an easy calculation shows that, when the
above conclusions are taken already into account, all equations that
can be derived from $R$ are automatically satisfied. However, we
have enough information to conclude that there does not exist a
regular Lagrangian for the vertically rolling disk and its first
associated second-order system. Indeed, the determinant of the
multiplier matrix
\[
(g_{ij})=\left(\begin{array}{cccc} g_{11} & g_{12} & \lambda g_{23}
& \lambda g_{24}
\\ g_{12} & g_{22} & g_{23} & g_{24} \\ \lambda g_{23} & g_{23} & 0 & 0\\ \lambda g_{24} & g_{24} & 0 & 0
\end{array} \right),
\]
(with $\lambda=-\dot\theta/\dot\varphi$) clearly vanishes and this
is a violation of one of the first Helmholtz conditions.

Thus, to summarize the above results, for the nonholonomic free
particle (\ref{nhfp}), the knife edge on the plane (\ref{kep1}) and
the vertically rolling disk (\ref{vd1}), we conclude that there does
not exist a regular Lagrangian for their first associated
second-order system (\ref{choice1}).

\subsection{Lagrangians for the second associated second-order system}

In this section, we will investigate the inverse problem for the
second associated system,
\begin{equation} \label{sode}
\ddot r_{1} =0, \qquad \ddot q_a = \Xi_a(r_{1}) \dot q_a\dot r_{1}.
\end{equation}
In the $q_a$-equations, there is no sum over $a$, which is an index
that runs from $2$ to the dimension of the configuration space, and
with respect to the formulation of the inverse problem in section 3,
we have $f_{1}= 0$ and $f_a= \Xi_a{\dot q}_a\dot r_{1}$. Moreover,
one can easily compute that the only non-vanishing components of
$\Phi$ are now
\[
\Phi^a_{{1}} = - \onehalf \dot r_{1} {\dot q}_a (2 \Xi'_a -
\Xi_a^2),\quad \Phi^a_a= \onehalf {\dot r_{1}}^2(2 \Xi'_a -
\Xi_a^2).
\]
The $\Phi$-conditions turn out to be quite simple: if $\Phi^a_a\neq
0$, then
\begin{equation}\label{alg1}
{\dot q}_a g_{aa}=-\dot r_{1} g_{{1}a},
\end{equation}
 and if $\Phi^a_a\neq
\Phi^b_b$ for $a\neq b$, then \begin{equation}\label{alg2} g_{ab}=0.
\end{equation}
 These
restrictions on $\Phi$ lead to the assumptions that first $\Xi_a
\neq 0$ and $\Xi_a\neq 2/(C-r_{1})$, where $C$ is any constant,
second that $\Xi_a\neq\Xi_b$ and, formally, $\Xi_a-\Xi_b\neq
E_b/(C-\int E_b dr_{1} )$, where $E_b(r_{1})=\exp(\int 2 \Xi_b
dr_{1})$. Suppose for now that we are dealing with nonholonomic
systems (\ref{NHsode}) where this is the case. Then one can easily
show that all the other $\nabla\ldots\nabla\Phi$-conditions do not
contribute any new information, as well as that the $R$-condition is
automatically satisfied. Thus we should therefore turn our attention
to the $\nabla$-condition, which is a PDE. To simplify the
subsequent analysis though, we note that although the multipliers
$g_{ij}$ can in general be functions of all variables
$(r_{1},q_a,\dot r_{1},{\dot q}_a)$, in view of the symmetry of the
system we shall assume them to be, without loss of generality,
functions of $(r_{1},\dot r_{1},{\dot q}_a)$ only.

Now, by differentiating the algebraic conditions by $r_{1}$, $\dot r_{1}$ and
${\dot q}_a$, we get the additional conditions
\begin{eqnarray*} &&
{\dot q}_a g'_{aa}=-\dot r_{1} g'_{1a}\\ && g_{aa} + {\dot q}_a
g_{aaa} =-\dot r_{1} g_{1aa}, \quad {\dot q}_a g_{1aa}=-g_{1a}-\dot
r_{1} g_{11a}
\\ && g_{aab} = 0 = g_{{1}ab}, \quad \mbox{if $a\neq b$.}
\end{eqnarray*}
Finally, the $\nabla$-Helmholtz conditions are, with the above
already incorporated,
\begin{eqnarray*}&&
g'_{{1}{1}} + \sum_b \Xi_b( g_{{1}{1}b} {\dot q}_b   - g_{bb}
 \frac{{\dot q}_b^2}{{\dot r_{1}}^2}) =0,\\
&& g'_{aa} + \Xi_{a} (g_{aaa}{\dot q}_a + g_{aa})  =0.
 \end{eqnarray*}
In what follows we will implicitly assume everywhere that $\dot
r_{1}\neq 0$. As a consequence, the multipliers $(g_{ij})$ (and the
Lagrangians we may derive from it) will only be defined for $\dot
r_{1}\neq 0$

It is quite impossible to find the most general solution for
$(g_{ij})$ though. We will show that there is an interesting class
of solutions if we make the anszatz that $g_{bbb}=0$ for all $b$.
With that and with the above $g_{aab}=0$ in mind, we conclude that
all such $g_{bb}$ will depend only on possibly $r_{1}$ and $\dot
r_{1}$. Moreover, from the last $\nabla$-conditions we can determine
their dependency on the variable $r_{1}$. Since now
\[
g'_{bb} + g_{bb} {\Xi_{b}} =0,\] it follows that $g_{bb}(r_{1},\dot
r_{1}) = F_b(\dot r_{1})\exp(-\xi_b(r_{1}))$, where $\xi_b$ is such
that $\xi'_b=\Xi_b$ and where $F_{b}(\dot r_{1})$ is still to be
determined from the remaining conditions. From one of the above
conditions we get $g_{{1}bb}=-g_{bb}/\dot r_{1}$ (since
$g_{bbb}=0$), so
\[ \frac{dF_b}{d\dot r_{1}} = - \frac{F_b}{\dot
r_{1}},
\]
from which $F_b=C_b/\dot r_{1}$, with $C_b$ a constant, and thus
$g_{bb}= C_b \exp(-\xi_b)/ \dot r_{1}$. Therefore, from the
algebraic conditions, $g_{{1}b} = - (g_{bb}/\dot r_{1}){\dot q}_b =
- C_b \exp(-\xi_b) {\dot q}_b/ {\dot r_{1}}^2$, and thus
$g_{{1}{1}b}= 2 C_b {\dot q}_b \exp(-\xi_b)/ {\dot r_{1}}^3$. With
this, the first $\nabla$-condition becomes
\[
g'_{{1}{1}} +  \sum_b C_b \exp(-\xi_b) \xi'_b \frac{{\dot
q}_b^2}{{\dot r_{1}}^3} =0,
\]
and thus
\[
g_{{1}{1}} = \sum_b C_b \exp(-\xi_b)  \frac{{\dot q}_b^2}{{\dot
r_{1}}^3} + C(\dot r_{1},{\dot q}_b).
\]
Given that $g_{{1}{1}b}= 2 C_b {\dot q}_b \exp(-\xi_b)/ {\dot
r_{1}}^3$, we can now determine the ${\dot q}_b$-dependence of $C$.
We simply get
\[
g_{{1}{1}} = \sum_b C_b \exp(-\xi_b)  \frac{{\dot q}_b^2}{{\dot
r_{1}}^3} + F_{{1}}(\dot r_{1}).
\]
Notice that $g_{{1}{1}{1}}$ does not show up explicitly in the
conditions or in the derived conditions. Therefore, there will
always be some freedom in the $g_{{1}{1}}$-part of the Hessian,
represented here by the undetermined function $F_{{1}}({\dot
r}_{1})$.

Up to a total time derivative, the most general Lagrangian whose
Hessian $g_{ij} = \spd{L}{{\dot q}^i}{{\dot q}^j}$ is the above
multiplier, is:
\begin{equation}\label{Lag1}
L = \rho(\dot r_{1}) + \onehalf \sum_b C_b \exp(-\xi_b) \frac{{\dot
q}_b^2}{\dot r_{1}},
\end{equation}
where $d^2\rho/d{\dot r_{1}}^2=F_{{1}}$. One can easily check that
the Lagrangian is regular, as long as $d^2\rho/d{\dot r_{1}}^2$ is
not zero, and as long as none of the $C_b$ are zero. Remark,
finally, that the Lagrangian is only defined on the whole tangent
space if $C_b=0$ (and $\rho$ is at least $C^2$ everywhere). We can
therefore only conclude that there is a regular Lagrangian (with the
ansatz $g_{bbb}=0$) on that part of the tangent manifold where $\dot
r_{1}\neq 0$. As a consequence, the solution set of the
Euler-Lagrange equations of the Lagrangian (\ref{Lag1}) will not
include those solutions of the second-order system (\ref{system1})
where $\dot r_{1} = 0$. In case of the vertically rolling disk, for
example, these solutions are exactly the special ones given by
(\ref{vd2}), and that is the reason why we will exclude them from
our formalism.

Recall that at the beginning of this section, we have made the
assumptions that $\Phi^a_a\neq 0$ and $\Phi^a_a\neq \Phi^b_b$.
Suppose now that one of these assumptions is not valid, say
$\Xi_2=0$ and therefore $\Phi^2_2=0$. Then, among the algebraic
Helmholtz conditions there will no longer be a relation in
(\ref{alg1}) that links $g_{22}$ to $g_{{1}2}$. In fact, since the
$g_{ij}$ now need to satisfy only a smaller number of algebraic
conditions, the set of possible Lagrangians may be larger. We can,
of course, still take the relation
\begin{equation}\label{ans1} {\dot q}_2 g_{22} = -\dot
r_{1} g_{{1}2} \end{equation}
 as an extra ansatz
(rather than as a condition) and see whether there exists
Lagrangians with that property. By following the same reasoning as
before, we easily conclude that the function (\ref{Lag1}) is also a
Lagrangian for systems with $\Phi^2_2=0$. In fact, it will be a
Lagrangian if any of the assumptions is not valid.

Apart from (\ref{ans1}), we are, of course, free to take any other
ansatz on $g_{12}$ and $g_{22}$. If we simply set
\[
g_{{1}2}=0,
\]
 it can easily be verified that also
\begin{equation}\label{Lag2}
L = \rho({\dot r}_1) + \sigma({\dot r}_2)+\onehalf\left(
\sum_{\alpha} a_\alpha \exp(-\xi_\alpha) \frac{{\dot
s}_\alpha^2}{{\dot r}_1}\right)
\end{equation}
 is a Lagrangian for a system
(\ref{sode}) with $\Xi_2=0$ (where, as usual,
$(q_a)=(r_2,s_\alpha)$). It is regular as long as both
$d^2\rho/d{\dot r}_1^2$ and $d^2\sigma/d{\dot r}_2^2$ do not vanish.

\begin{propo} The function
 \begin{equation}\label{L2as} L = \rho(\dot
r_{1}) + \frac{1}{2N} \left( C_2 \frac{{\dot r_{2}}^2}{\dot r_{1}}
+\sum_\beta C_\beta \frac{{\dot s}_\beta^2}{A_\beta \dot
r_{1}}\right),
\end{equation}
with $d^2\rho/d{\dot r_{1}}^2\neq 0$ and all $C_\alpha\neq 0$ is in
any case a regular Lagrangian for the second associated systems
(\ref{NHsode}). If the invariant measure density $N$ is a constant,
then also
 \begin{equation}\label{L2as2} L = \rho(\dot
r_{1}) + \sigma(\dot r_{2})+ \frac{1}{2 N} \sum_\beta a_\beta
\frac{{\dot s}_\beta^2}{A_\beta \dot r_{1}},
\end{equation}
where $d^2\rho/d{\dot r_{1}}^2\neq 0$, $d^2\sigma/d{\dot
r_{1}}^2\neq 0$ and all $C_\alpha\neq 0$ is a regular Lagrangian for
the second associated systems (\ref{NHsode}).
\end{propo}
\begin{proof} For the second associated systems, the second-order
equations (\ref{sode}) are of the form (\ref{NHsode}). One easily
verifies that in that case
\begin{equation}\label{nonholxi}
\xi_{{2}} = \ln N \quad \mbox{and} \quad \xi_\alpha = \ln(N
A_\alpha)
\end{equation}
are such that $\xi'_a=\Xi_a$. The first Lagrangian in the theorem is
then equal to the one in (\ref{Lag1}). For a system with constant
invariant measure $N$, we get that $\Xi_2=0$. Therefore, also the
function (\ref{Lag2}) is a valid Lagrangian.  \end{proof}

Let us end this section with a list of the Lagrangians for the
nonholonomic free particle, the knife edge on a horizontal plane and
the vertically rolling disk. The respective Lagrangians (\ref{L2as})
for the first two examples are:
\begin{equation}\label{nhpL2as} L =
\rho(\dot x) + \onehalf \sqrt{1+x^2} \left( C_2 \frac{{\dot
y}^2}{\dot x} +C_3\frac{{\dot z}^2}{x \dot x }
  \right),
\end{equation}
and
\begin{eqnarray}
L &=& \rho(\dot \phi) + \onehalf\sqrt{m(1+\tan(\phi)^2)} \left(C_2
\frac{{\dot x}^2}{\dot \phi} +C_3\frac{{\dot y}^2}{\tan(\phi) \dot
\phi }\right), \nonumber \\
&=& \rho(\dot \phi) +\onehalf C_2 \sqrt{m} \frac{{\dot
x}^2}{\cos(\phi) \dot \phi} +\onehalf  C_3 \sqrt{m}\frac{{\dot
y}^2}{\sin(\phi) \dot \phi }. \label{keL2as}
\end{eqnarray}
The vertically rolling disk in one of those systems with constant
invariant measure. The first Lagrangian (\ref{L2as}) is:
\begin{equation}\label{vdL2as1} L= \rho(\dot \varphi) + \frac{\sqrt{I +
mR^2}}{2} \left (C_2 \frac{{\dot \theta}^2}{\dot \varphi}
+C_3\frac{{\dot x}^2}{\cos(\varphi) \dot \varphi} + C_4 \frac{{\dot
y}^2}{\sin(\varphi) \dot \varphi}\right)
\end{equation}
and the second Lagrangian (\ref{L2as2}) is:
\begin{equation}\label{vdL2as2} L= \rho(\dot
\varphi) + \sigma({\dot \theta})
-\frac{\sqrt{I+mR^2}}{2}\left(a_3\frac{{\dot x}^2}{\cos(\varphi)
\dot \varphi} + a_4 \frac{{\dot y}^2}{\sin(\varphi) \dot
\varphi}\right).
\end{equation}


\subsection{Lagrangians for the third associated second-order system}

In section 2.1 we have described a third associated second-order
system (\ref{choice3VRD}) for the example of the vertically rolling
disk. That system comes actually from a comparison of the
variational nonholonomic and the Lagrange-d'Alembert nonholonomic
equations of motion we conducted in \cite{FB}. There we investigated
the conditions under which the variational nonholonomic Lagrangian
$L_V$ would reproduce the nonholonomic equations of motion when
restricted to the nonholonomic constraint manifold. Thus, instead of
associating second-order systems to nonholonomic equations and
applying the techniques of the inverse problem to derive the
Lagrangian (and the Hamiltonian), in \cite{FB} we started from a
specific Lagrangian (the variational nonholonomic Lagrangian $L_V$)
and investigated the conditions under which its variational
equations match the nonholonomic equations. Other relevant work on
this matter can be found in e.g. \cite{CDMM}.

In case of our class of nonholonomic systems with Lagrangian
(\ref{LEucl}) and constraints (\ref{constraints}) the variational
nonholonomic Lagrangian is simply
\begin{eqnarray*}\label{LV1} L_{V} &=& L -\sum_\alpha \frac{\partial L}{\partial
\dot{s}_{\alpha}}(\dot{s}_{\alpha}+A_{\alpha} \dot{r}_{2})\\ &=&
\onehalf (I_1 {\dot r}_1^2 + I_1 {\dot r}_1^2 - \sum_\alpha I_\alpha
{\dot s}_\alpha^2 ) - \sum_\alpha A_\alpha I_\alpha {\dot s}_\alpha
{\dot r}_2.
\end{eqnarray*}
A short calculation shows that its Euler-Lagrange equations in
normal form are given by
\begin{eqnarray} \label{choice3}
 &&\ddot r_{1} = - \big(\sum_\beta I_\beta A'_\beta {\dot s}_\beta\big) {\dot r}_2 , \qquad \qquad \ddot r_{2} = -N^2
\big(\sum_\beta I_\beta A_\beta A'_\beta\big) \dot r_{1}\dot r_{2} + \big(\sum_\beta I_\beta A'_\beta {\dot s}_\beta\big) {\dot r}_1,\nonumber\\
&& {\ddot s}_\alpha = -\Big(A'_{\alpha} - N^2 A_\alpha
\big(\sum_\beta I_\beta A_\beta A'_\beta\big)\Big) \dot r_{1} \dot
r_{2} - A_\alpha (\sum_\beta I_\beta A'_\beta {\dot s}_\beta) {\dot
r}_1.
\end{eqnarray}
In general, these systems are not associated to our class of
nonholonomic systems. That is, the restriction of their solutions to
the constraint manifold ${\dot s}_\alpha =- A_\alpha {\dot r}_2$ are
not necessarily solutions of the nonholonomic equations
(\ref{NHeq}). However, in case that the invariant measure density
$N$ is a constant, we have that $\sum_\beta I_\beta A_\beta
A'_\beta=0$. As a consequence,  all the terms in the equations
(\ref{choice3}) that contain $\sum_\beta I_\beta A'_\beta {\dot
s}_\beta$ vanish when we restrict those equations to the constraint
manifold and the equations in ${\ddot s}_\alpha$ integrate to the
equations of constraint (\ref{constraints}). The restriction of the
equations (\ref{choice3}) is therefore equivalent with the
nonholonomic equations ({\ref{NHeq}}). We conclude the following.

\begin{propo} If $N$ is constant, the equations
(\ref{choice3}) form an associated second-order system and, by
construction, they are equivalent to the Euler-Lagrange equations of
the variational nonholonomic Lagrangian $L_V$.
\end{propo}

We refer to \cite{FB} for more details and some more general
statements on this way of finding a Lagrangian for a nonholonomic
system and we end the discussion on the third associated systems
here.


\section{Hamiltonian formulation and the constraints in phase space}

In the situations where we have found a regular Lagrangian, the
Legendre transformation leads to an associated Hamiltonian system.
Since the base solutions of the Euler-Lagrange equations of a
regular Lagrangian are also base solutions of Hamilton's equations
of the corresponding Hamiltonian, the Legendre transformation $FL$
will map those solutions of the Euler-Lagrange equations that lie in
the constraint distribution ${\mathcal D}$ to solutions of the
Hamilton equations that belong to the constraint manifold ${\mathcal
C} = FL({\mathcal D})$ in phase space. Recall however that the
Lagrangians for the second associated second-order systems (and
their Legendre transformation) were not defined on ${\dot r}_1 =0$,
and so will also the corresponding Hamiltonians.

Let us put for convenience $\rho(\dot r_{1})=\frac{1}{2}I_1{\dot
r_{1}}^2$ and $\sigma({\dot r}_2)=\onehalf I_2 {\dot r}_2^2$ in the
Lagrangians of the previous section.
\begin{propo}
Given the second associated second-order system (\ref{NHsode}), the
regular Lagrangian (\ref{L2as}) (away from $\dot r_{1}=0$) and
constraints (\ref{constraints}) on $TQ$ are mapped by the Legendre
transform to the Hamiltonian and constraints in $T^*Q$ given by:
\begin{equation}\label{Ham11} H= \frac{1}{2I_1} \left(p_{{1}}+
\onehalf N \left( \frac{p_{{2}}^2}{C_2} + \sum_\beta A_\beta
\frac{p_\beta^2}{C_\beta} \right)\right)^2, \qquad
C_2p_\alpha=-C_\alpha p_{{2}}.
\end{equation}
In case $N$ is constant, the second Lagrangian (\ref{L2as2}) and
constraints (\ref{constraints}) are transformed into
\begin{equation}\label{Ham2}
H= \frac{1}{2I_2} p_2^2 + \frac{1}{2I_1} \left(p_1+ \onehalf {N}
\left( \sum_\beta \frac{A_\beta}{a_\beta} {p_\beta^2}
\right)\right)^2, \quad I_2{ {N}{\dot r}_1 p_\alpha} +{a_\alpha}
p_2=0,
\end{equation}
where ${\dot r}_1(r_1,p_1,p_\alpha)=(p_1 + \onehalf {N} \sum_\alpha
A_\alpha p_\alpha^2/a_\alpha)/I_1$.
\end{propo}
\begin{proof} The Legendre transformation gives for the Lagrangian
(\ref{Lag1})
\begin{equation}\label{mom}
p_{1}= I_1 \dot r_{1}-\onehalf\sum_b C_b\exp(-\xi_b)\frac{{\dot
q}_b^2}{{\dot r_{1}}^2},\qquad p_b = C_b\exp(-\xi_b)\frac{{\dot
q}_b}{\dot r_{1}},
\end{equation}
from which one can easily verify that the corresponding Hamiltonian
is
\begin{equation}
H= \frac{1}{2I_1} \left(p_{{1}}+\onehalf \sum_b \exp(\xi_b)
\frac{p_b^2}{C_b} \right)^2. \label{Ham1}
\end{equation}
In the case of the second associated second-order systems in the
form (\ref{NHsode}), the $\xi_a$ take the form (\ref{nonholxi}), and
we obtain the Hamiltonian in expression (\ref{Ham11}). From
(\ref{mom}) we can then compute the constraint manifold ${\cal C}$
in phase space. Since now
\[ p_{{2}} = C_2\frac{\dot r_{2}}{{N}\dot r_{1}}
\quad\mbox{and}\quad p_\alpha = C_\alpha\frac{{\dot
q}_\alpha}{{N}\dot r_{1}},
\]
the constraints (\ref{constraints}) can be rewritten as
\[
\dot r_{1} \left(\frac{p_\alpha}{C_\alpha} +
\frac{p_{2}}{C_2}\right)=0,
\]
where $\dot r_{1} = \frac{1}{I_1} (p_{{1}} + \onehalf {N}
(p_{{2}}{\dot r}_2^2/C_2 +\sum_\beta A_\beta p_\beta^2/C_\beta))$.
Assuming as always that $\dot r_{1}\neq 0$, we get that the
constraint manifold in phase space is given by
$C_2p_\alpha=-C_\alpha p_{{2}}$ for all $\alpha$.

An analogous calculation with the Lagrangian (\ref{L2as2}) gives the
Hamiltonian and the constraints in (\ref{Ham2}), in the case where
$N$ is constant.
\end{proof}

We can recover the Hamiltonians of \cite{QB} from Proposition 2. As
perhaps the simplest example, note that with
$(r_{1},r_{2},s_{\alpha})=(x,y,z)$, by taking $C_2$ and $C_3$ both
to be 1, and $A(r_{1})=x$, we recover the Hamiltonian and the
constraint that appears in \cite{QB} for the nonholonomic free
particle.

Consider now the knife edge on the plane. Taking
$C_2=C_3=1/\sqrt{m}$ and $A(\phi)=-\tan(\phi)$ gives:
\begin{equation}\label{Hknife}
H= \frac{1}{2J} \left(p_\phi + \onehalf (\cos(\phi) p_x^2 -
\sin(\phi) p_y^2)\right)^2,
\end{equation}
while the constraint manifold becomes
\begin{equation}\label{Cknife}
p_x+p_y=0.
\end{equation}

For the rolling disk we get for the first Hamiltonian (\ref{Ham11})
\[
H=\frac{1}{2J} \left( p_\varphi +\frac{1}{2\sqrt{I + mR^2}} \left(
\frac{p_\theta^2}{C_2} - \frac{\cos(\varphi)p_x^2}{ C_3} -
\frac{\sin(\varphi)p_y^2}{C_4} \right)\right)^2,
\]
and $C_2p_x=-C_3 p_\theta$ and $C_2p_y=-C_4 p_\theta$ for the
constraints. These are not the Hamiltonian and the constraints that
appear in \cite{QB} though. It turns out that the Hamiltonian and
the constraints in \cite{QB} are in fact those that are associated
to the second Hamiltonian (\ref{Ham2}). It is, with, for example,
$a_3=a_4=-J/\sqrt{I+mR^2}$ of the form
\[
H=\frac{1}{2I} p_\theta^2 + \frac{1}{2}\left(p_\varphi + \onehalf
p_x^2\cos(\varphi) + \onehalf p_y^2\sin(\varphi)\right)^2
\]
and the constraints are
\[
\dot \varphi p_x =p_\theta,\qquad \dot \varphi p_y = p_\theta
\]
where $\dot \varphi = p_\varphi + \onehalf \cos(\varphi) p_x^2 +
\onehalf \sin(\varphi) p_y^2$ or, equivalently,
\[
p_x - p_y =0, \qquad \dot \varphi p_x - p_\theta =0,
\]
as the constraints appears in \cite{QB}.


\section{Pontryagin's Maximum Principle}

Consider the optimal control problem of finding the controls $u$
that minimize a given cost function $G(x,u)$ under the constraint of
a first order controlled system $\dot x = f(x,u)$. One of the
hallmarks of continuous optimal control problems is that, under
certain regularity assumptions, the optimal Hamiltonian can be found
by applying the Pontryagin maximum principle. Moreover, in most
cases of physical interest, the problem can be rephrased so as to be
solved by using Lagrange multipliers $p$. Form the Hamiltonian
$H^P(x,p,u) = \langle p, f(x,u) \rangle - p_0 G(x,u)$ and calculate,
if possible, the function $u^*(x,p)$ that satisfies the optimality
conditions
\[
 \fpd{H^P}{u}(x,p,u^*(x,p)) \equiv 0.
\]
  Then, an extremal $x(t)$ of the optimal control problem
 is also a base solution of Hamilton's equations for the
optimal Hamiltonian given by $H^*(x,p)=H^P(x,p,u^*(x,p))$. The
optimal controls $u^*(t)$ then follow from substituting the
solutions $(x(t),p(t))$ of Hamilton's equations for $H^*$ into
$u^*(x,q)$.

Such a usage of the multiplier approach can also be applied with
succes to the mechanics of physical systems with holonomic
constraints. However, in the case of nonholonomically constrained
systems the Lagrange multiplier approach, also called the {\em
vakonomic} approach by Arnold \cite{Ar}, generally leads to dynamics
that do not reproduce the physical equations of motion (see
\cite{CDMM,Le} and references therein). Thus, the rich interplay
between Pontryagin's Maximum Principle, the vakonomic approach, and
the physical equations of motion of a constrained system breaks down
when the constraints are nonholonomic. However as we showed in a
previous paper \cite{FB}, for certain systems and initial data the
vakonomic approach and Lagrange-D'Alembert principle yield
equivalent equations of motion.

We will show here for the second associated systems
\[
\ddot r_{1} =0, \qquad \ddot q_a = \Xi_a(r_{1}) \dot q_a\dot r_{1},
\]
that we can also find the Hamiltonians of the previous section via a
rather ad hoc application of Pontryagin's Maximum Principle. Hereto,
let us put $\Xi_a=\xi'_a$ as before and observe that the above
second-order system can easily be solved for $(\dot r_{1}(t),{\dot
q}_a(t))$. Indeed, obviously $\dot r_{1}$ is constant along
solutions, say $u_{{1}}$. We will suppose as before that $u_{1}\neq
0$. From the $q_a$-equations it also follows that ${\dot q}_a/
\exp(\xi_a)$ is constant, and we will denote this constant by $u_a$.
To conclude,
\[ \dot r_{1}(t) = u_{{1}}, \quad {\dot
q}_a(t)= u_a \exp(\xi_a(r_{1}(t))).
\]
Keeping that in mind, we can consider the following {\em associated
controlled first-order system}
\begin{equation}\label{fode} \dot r_{1} =
u_{{1}}, \quad {\dot q}_a= u_a \exp(\xi_a(r_{1}))
\end{equation}
(no sum over $a$), where $(u_{{1}},u_a)$ are now interpreted as
controls.

The next proposition relates the Hamiltonians of Proposition 2 to
the optimal Hamiltonians for the optimal control problem of certain
cost functions, subject to the constraints given by the controlled
system (\ref{fode}).

\begin{propo} The optimal Hamiltonian $H^*$ of the optimal control problem of minimizing the cost function
\[ G_{1}(r_{1},q_a,u_{{1}},u_a)  =  \onehalf\left(I_1 u_{{1}}^2 +
\sum_a C_a \exp(\xi_a(r_{1})) \frac{u_a^2}{u_{{1}}}\right)
\]
subject to the dynamics (\ref{fode}) is given by:
\begin{equation}\label{h1} H^*(q,p)= \frac{1}{2I_1} \left( p_{r_{1}}
+ \onehalf \sum_b \exp(\xi_b) \frac{p_b^2}{C_b} \right)^2.
\end{equation}
If $\Xi_2$ is zero, the optimal Hamiltonian for the optimal control
problem of minimizing the cost function
\[
G_{2}(r_{1},q_a,u_{{1}},u_a)  = \onehalf\left(I_1 u_{{1}}^2 + I_2
u_2^2 + \sum_\alpha a_\alpha \exp(\xi_\alpha(r_{1}))
\frac{u_\alpha^2}{ u_{{1}}}\right),
\]
subject to the dynamics (\ref{fode}) is given by:
\begin{equation}\label{h2} H^*(q,p) = \frac{1}{2I_2} p_{2}^2 +
\frac{1}{2I_1} \left(  p_{{1}} + \onehalf \sum_\beta \exp(\xi_\beta)
\frac{p_\beta^2}{a_\beta} \right)^2.
\end{equation}
In case the controlled system is associated to a nonholonomic system
(that is, in case the $\xi_a$ take the form (\ref{nonholxi})), the
above Hamiltonians are respectively the Hamiltonians (\ref{Ham11})
and (\ref{Ham2}) of Proposition 2.
\end{propo}

\begin{proof} The Hamiltonian $H^P$ is
\begin{equation}\label{Ha}
H^P(r_{1},q_a,p_{{1}},p_a,u_{{1}},u_a)= p_{{1}} u_{{1}} + \sum_a p_a
u_a
 \exp(\xi_a) - G_1.
\end{equation}
The optimality conditions $\partial H^P/\partial u_{{1}}=0$,
$\partial H^P/\partial u_{a}=0$, together with the assumption that
$u_{{1}}\neq 0$, yield the following optimal controls as functions
of $(q,p)$:
\begin{eqnarray*}
I_1 u_{{1}}^* & = & p_{{1}} + \frac{1}{2}\sum_a \exp(\xi_a) \frac{p_a^2}{C_a}, \label{ua} \\
\frac{u_a^*}{u_{{1}}^*} & = & \frac{p_a}{C_a}.
\end{eqnarray*}
For the Hamiltonian $H^*(q,p)=H^P(q,p,u^*(q,p))$, we get
\begin{eqnarray*}
H^*(q,p) &=& \left(p_1 - \onehalf I_1 u^*_1\right) u^*_1 + \sum_a
\exp(\xi_a)u^*_a \left(p_a -\onehalf C_a \frac{u^*_a}{u^*_1}\right) \\
&=& \frac{1}{I_x} \left[
 \left(\onehalf p_1 - \frac{1}{4} \sum_a\exp(\xi_a) \frac{p^2_a}{C_a}\right)\left(p_1 + \onehalf \sum_b
 \exp(\xi_b)\frac{p^2_b}{C_b}\right)\right.
\\ &&\qquad\qquad\qquad\qquad+ \left. \onehalf \sum_a \exp(\xi_a) \frac{p_a^2}{C_a} \left(p_1
+\onehalf \sum_b\exp(\xi_b) \frac{p^2_b}{C_b}\right)   \right]
\\ &=& \frac{1}{2I_1} \left(   p_{{1}} + \onehalf \sum_b
\exp(\xi_b) \frac{p_b^2}{C_b} \right)^2,
\end{eqnarray*}
which is exactly the Hamiltonian (\ref{Ham1}).

For the second cost function, with $\Xi_2=0$, we get for
Pontryagin's Hamiltonian
\[
H^P= p_{r_{1}} u_{r_{1}} + p_{r_{2}} u_{r_{2}} + \sum_\alpha
p_\alpha u_\alpha \exp(\xi_\alpha) - G_2.
\]
The optimal controls as functions of $(q,p)$ are now
\begin{eqnarray*}
I_1 u_{r_{1}}^* & = & p_{r_{1}} + \frac{1}{2}\sum_\alpha \exp(\xi_\alpha) \frac{p_\alpha^2}{a_\alpha},\\
I_2 u_{r_{2}}^* &=& p_{r_{2}},\\ \frac{u_\alpha^*}{u_{r_{1}}^*} & = &
\frac{p_\alpha}{C_\alpha}.
\end{eqnarray*}
With this the Hamiltonian becomes
\[
H^*(q,p)= \frac{1}{2I_2} p_{r_{2}}^2 + \frac{1}{2I_1} \left(  p_{r_{1}} +
\onehalf \sum_\beta \exp(\xi_\beta) \frac{p_\beta^2}{a_\beta}
\right)^2,
\]
which is exactly (\ref{Ham2}) after the substitution (\ref{nonholxi}).
\end{proof}


\section{Related Research Directions and Conclusions}

In essence, the method we have introduced in the previous sections
resulted in an unconstrained, variational system which when
restricted to an appropriate submanifold reproduces the dynamics of
the underlying nonholonomic system. Although we have restricted our
attention to a certain explicit subclass of nonholonomic systems,
many of the more geometric aspects of the introduced method seem to
open the door to generalizing the results to larger classes of
systems. For example, a lot of what has been discussed was in fact
related to the somehow hidden symmetry of the system. That is to
say: both the Lagrangian (\ref{LEucl}) and the constraints
(\ref{constraints}) of the systems at hand were explicitly
independent of the coordinates $r_2$ and $s_\alpha$. This property
facilitated the reasoning we have applied in our study of the
corresponding inverse problems. One possible path to the extension
of some of the results in this paper may be the consideration of
systems with more general (possibly non-Abelian) symmetry groups. A
recent study \cite{MikeTom} of the reduction of the invariant
inverse problem for invariant Lagrangians may be helpful in that
respect.

The methods of the inverse problem have lead us to the Lagrangians
for the second associated second-order systems (\ref{sode}). For
those systems the $q^a$-equations were, apart from the coupling with
the $r_{1}$-equation, all decoupled from each other. It would be of
interest to see, for more general systems, how such a form of
partial decoupling influences the question of whether or not a
regular Lagrangian exists.

In the previous section we have found a new link between the fields
of optimal control, where equations are derived from a Hamiltonian,
and nonholonomic mechanics, where equations are derived from a
Hamiltonian and constraint reaction forces. By combining elements of
both derivations, for certain systems one can formulate the
mechanics in a form analogous to the treatment of constraints
arising from singular Lagrangians that leads to the Dirac theory of
constraints \cite{Q}, which allows for the quantization of
constrained systems wherein the constraints typically arise from a
singular Lagrangian (see \cite{Kl} and references therein). Central
to the method is the modification of the Hamiltonian to incorporate
so-called
 first and second class constraints.

The method proposed in this paper in a sense provides an analogue to
Dirac's theory and allows for the investigation into the
quantization of certain nonholonomic systems by similarly modifying
the usual Hamiltonian. In attempting to quantize the class of
systems we have considered, we can now instead use one of the
Hamiltonians found in Proposition 3. We should note that there have
already been some attempts to quantize nonholonomic systems
\cite{BR,Ed,QB,HP,Kl,B2}, and that the results have been mixed,
mainly due to the inherent difficulties arising in the quantization
procedure, as well as the difficulties in dealing with the system's
constraints. However, the present work enables one to treat the
constraints more like an initial condition, since, for example, the
constraint (\ref{Cknife}) is really the relation $c_{1}+c_{2}=0$,
where $p_{x}=c_{1}$, and $p_{y}=c_{2}$ follows from $H$ in
(\ref{Hknife}) being cyclic in $x,y$. Such a treatment of the
constraints eliminates much of the difficulty arising in attempting
to quantize some nonholonomic systems.

As example,  consider the knife edge on the plane, in view of
(\ref{Hknife}). We can take the quantum Hamiltonian $\hat{H}$ to be
of the form
\begin{equation}
\hat{H} = -\frac{\hbar^{2}}{2}\left[\frac{\partial}{\partial
\phi}-i\frac{\hbar}{2}\left(\cos(\phi)\frac{\partial^{2}}{\partial
x^{2}}-\sin(\phi)\frac{\partial^{2}}{\partial
y^{2}}\right)\right]^{2}, \label{qkep}
\end{equation}
which is a Hermitian operator, and consider the quantum version of
the constraint manifold (\ref{Cknife}) $\hat{p}_{x}+\hat{p}_{y}=0$
(Here $\hat{}\,$ stands for the quantum operator form under
canonical quantization). There have in the literature been
essentially two different ways to impose these quantum constraints:
{\em strongly} and {\em weakly}. One may require that the quantum
constraints hold {\em strongly}, by restricting the set of possible
eigenstates of (\ref{qkep}) to those which satisfy the quantum
operator form of constraint (\ref{Cknife}). On the other hand, one
may only require that the eigenstates $|\psi\rangle$ of (\ref{qkep})
satisfy the quantum constraints in mean
\begin{equation}
\langle\psi|\hat{p}_{x}+\hat{p}_{y}|\psi\rangle =0, \label{cim}
\end{equation}
a weaker condition but arguably a more physically relevant viewpoint
also advocated in \cite{QB,HP}. In \cite{QB} the authors show that
(using the example of the vertical rolling disk) the weaker version
can be used to recover the classical nonholonomic motion in the
semi-classical limit of the quantum dynamics. Details of our
application of the methods in this paper to the quantization of
nonholonomic systems will be presented in a future publication.


\section*{Acknowledgements}
The research of AMB and OEF was supported in part by the Rackham
Graduate School of the University of Michigan, through the Rackham
Science award, and through NSF grants DMS-0604307 and CMS-0408542.
TM acknowledges a Marie Curie Fellowship within the
 6th European Community Framework Programme and a research grant of the Research Foundation - Flanders.
 We would like to thank Alejandro Uribe and Willy Sarlet for helpful discussions.


\end{document}